\documentclass[letterpaper,12pt]{article}
\pdfoutput=1

\usepackage{amssymb,amsmath,amsthm,amsfonts}
\usepackage{graphicx, xcolor, varwidth}
\usepackage[percent]{overpic}
\usepackage{mathtools}
\usepackage{setspace}
\usepackage{cite}
\usepackage{tikz,tkz-graph,tkz-berge}
\usepackage[shortlabels]{enumitem}

\usetikzlibrary{decorations.markings}
\usetikzlibrary{shapes.geometric}
\tikzstyle{none}=[inner sep=0pt] 
\tikzstyle{simplethree}=[circle,fill=black,draw=black,line width=1.500,minimum size = 2pt,inner sep = 2pt]
\tikzstyle{simple}=[circle,fill=black,draw=black,line width=1.000,minimum size = 2pt,inner sep = 1.5pt]
\tikzstyle{simpletoo}=[circle,fill=black,draw=black,line width=1.000,minimum size = 1pt,inner sep = 1pt]

\usepackage{xcolor}
\usepackage[dotinlabels]{titletoc}
\usepackage[colorlinks=true, urlcolor=dblue, linktocpage=true, linkcolor=dblue,citecolor=purple]{hyperref}

\newcommand{\Beta}{\text{B}}

\newtheorem{defn}{Definition}
\newtheorem{fact}{Fact}
\newtheorem{comm}{Comment}

\newtheorem{lemma}{Lemma}

\newcommand{\mrm}{\mathrm}

\newcommand{\nn}{\nonumber}
\newcommand{\pd}{\partial}
\newcommand{\lb}{\left(}
\newcommand{\rb}{\right)}

\newcommand{\FF}{\mathcal{F}}

\newcommand{\MM}{\mathcal{M}}

\newcommand{\lsb}{\left[}
\newcommand{\rsb}{\right]}

\newcommand\be{\begin{equation}}
\newcommand\ba{\begin{eqnarray}}
\newcommand\ee{\end{equation}}
\newcommand\ea{\end{eqnarray}}

\numberwithin{equation}{section}

\DeclareMathOperator{\sgn}{sgn}
\newcommand\Qp{{\mathbb{Q}_p}}
\newcommand\Zp{{\mathbb{Z}_p}}

\colorlet{dblue}{blue!70!black}

\newcommand{\arxivold}[1]
  {\href{http://arxiv.org/abs/#1}{#1}}
\newcommand{\arxiv}[1]
  {\href{http://arxiv.org/abs/#1}{arXiv:#1}}

\usepackage{graphicx}
\usepackage{amsmath}

\begin{document}
\begin{spacing}{1.3}
\begin{titlepage}

\ \\
\vspace{-2.3cm}
\begin{center}

\begin{spacing}{2.3}
{\LARGE{From $p$-adic to Archimedean Physics: Renormalization Group Flow and Berkovich Spaces}}
\end{spacing}

\vspace{0.5cm}
An Huang,$^{1}$ Dan Mao,$^{1,3}$ and Bogdan Stoica$^{2}$

\vspace{5mm}

{\small

\textit{
$^1$Department of Mathematics, Brandeis University, Waltham, MA 02453, USA}\\

\vspace{2mm}

\textit{
$^2$Center of Mathematical Sciences And Applications, Harvard University, Cambridge MA 02138, USA }\\

\vspace{2mm}

\textit{
$^3$Department of Mathematics, University of Oxford, Oxford, OX2 6GG, UK}\\

\vspace{4mm}

{\tt anhuang@brandeis.edu, ella1120@brandeis.edu, bstoica@cmsa.fas.harvard.edu}

\vspace{0.5cm}
}

\end{center}

\begin{abstract}
\noindent We introduce the $p$-adic particle-in-a-box as a free particle with periodic boundary conditions in the $p$-adic spatial domain. We compute its energy spectrum, and show that the spectrum of the Archimedean particle-in-a-box can be recovered from the $p$-adic spectrum via an Euler product formula. This product formula arises from a flow equation in Berkovich space, which we interpret as a space of theories connected by a kind of renormalization group flow. We propose that Berkovich spaces can be used to relate $p$-adic and Archimedean quantities~generally.
\end{abstract}

\vspace{4cm}

\begin{center}
\emph{In memory of Steven Gubser.}
\end{center}

\end{titlepage}

\setcounter{tocdepth}{2}
\tableofcontents

\newpage

\section{Introduction}
\label{intro}

A central question in modern $p$-adic physics is how the usual Archimedean world emerges from the non-Archimedean worlds indexed by primes. This question has been of interest since the advent of $p$-adic string theory, starting with the works of Freund and Olson \cite{FreundOlson}, and Freund and Witten \cite{FreundWitten}, which showed that the $p$-adic four-point Veneziano amplitudes for open strings product into the multiplicative  inverse of the Archimedean four-point Veneziano amplitude, at tree level. There is by now a vast library of objects which obey such a product rule; examples include quantum mechanics, field theory and holography propagators, on-shell partition functions in two-dimensional gravity, Green's functions, and so on. However, it has also been known for a long time that many quantities do not obey simple product rules. For instance, \cite{BFOW,adelicNpoint} have shown that the higher point Veneziano amplitudes do not obey a product formula, away from special hyperplanes in the kinematic parameters.

The situation is complicated by the existence in the literature of other ways of relating the $p$-adic and Archimedean worlds. It was noticed in \cite{gershata} that analytically continuing in the prime parameter $p$ and taking the $p\to1$ limit in the effective Lagrangian for the p-adic tachyonic amplitudes produces an Archimedean Lagrangian for an effective tachyonic theory with logarithmic potential, called the Gerasimov-Shatashvili Lagrangian. Naively taking $p\to1$ may seem as an ad-hoc prescription, however \cite{Bocardo-Gaspar:2017atv,bggczg} have shown how to reproduce the four-point and five-point scattering amplitudes of the Gerasimov-Shatashvili Lagrangian from $p$-adic amplitudes, in a rigorous manner, using the $p\to1$ limit, and have conjectured that their procedure works for all $N$-point tree level amplitudes.

In this paper we would like to take the first steps toward systematically understanding the emergence of the Archimedean world from $p$-adic physics, along the lines of the program proposed in \cite{Stoica:2018zmi}. We will do this by studying a simple problem, which is the spectrum of the $p$-adic particle in a box; this is a simple example for which the ideas that will be introduced below can be explained cleanly. In Section \ref{freeparticle} we will solve the spectrum of the $p$-adic particle-in-a-box, and we will find that the Archimedean spectrum can be reconstructed from the $p$-adic one, via a product formula. In Section~\ref{secBerk} we will introduce the Berkovich space $\MM(\mathbb{Z})$, and we will interpret Berkovich spaces in general as the natural settings for a kind of renormalization group flow. While the usual renormalization group flow is with respect to energy scale, the Berkovich flow is with respect to a parameter controlling the norm of the space on which the theory is defined. Along a certain path in Berkovich space this parameter can be thought of as a proxy for energy scale, however in general Berkovich flow is more complicated than flow along a single direction (such as an energy scale). In the case of the space $\MM\lb\mathbb{Z}\rb$, the $p$-adic branches can be thought of as providing the setting for a $p$-adic ultraviolet completion of the Archimedean theory. We furthermore interpret the Euler product obeyed by the spectrum as arising from a flow equation in Berkovich space, and we propose that Berkovich spaces can be used to systematically understand the reconstruction of Archimedean quantities from $p$-adic ones in the sense of \cite{Stoica:2018zmi}. This includes quantities which do not obey simple product formulas, and the $p\to1$ analytic continuation procedure mentioned above. We conclude in Section \ref{sec4} with an analogy between flow in Berkovich space and general covariance in general~relativity.

Other recent works which have recently explored $p$-adic directions include \cite{Gubser:2016guj, Heydeman:2016ldy, Gubser:2016htz, Gubser:2017vgc, Bhattacharyya:2017aly, Gubser:2017tsi, Dutta:2017bja, Gubser:2017qed, Marcolli:2018ohd, Gubser:2018bpe, Gubser:2018yec, Qu:2018ned, Jepsen:2018dqp, Gubser:2018cha, Heydeman:2018qty, Hung:2018mcn, Jepsen:2018ajn, Gubser:2019uyf, Garcia-Compean:2019jvk, Huang:2019nog}.

\section{$p$-adic particle-in-a-box}
\label{freeparticle}

In this section, we will discuss the $p$-adic version of the quantum particle in a box. Physical theories defined on objects such as $\mathbb{R}$ or smooth manifolds are intimately connected to the topology of the ``points'' of these underlying objects. This can be seen straight-up in the way the theories are defined, for instance in the central role derivatives, which depend strongly on topology, play in Hamiltonians and the equations of motion. The punchline to introducing quantum mechanics (or indeed any other type of theory) on an object such as $\Qp$ is that one needs to ``isolate'' the dependence of the physical theory on the point set topology of the underlying object, and then to replace it with a different topology. For quantum mechanics this can be done by introducing an object known as the Vladimirov derivative $\pd_s$, which is the inverse Fourier transform of multiplication by $|k|^s$ in momentum space, with $k$ the momentum (this point of view on the derivative is also common in quantum mechanics on $\mathbb{R}$). The theory on $\Qp$ can then be built in terms of the Vladimirov derivative, in the same way quantum mechanics on $\mathbb{R}$ is built from the usual derivatives.

\subsection{Definitions}

The first task is to introduce Fourier transforms for functions $\psi:\Qp\to \mathbb{C}$. This is done by introducing additive characters on $\Qp$, which are the complex exponentials entering Fourier transforms.
\begin{defn}
An additive character $\chi:\Qp\to\mathbb{C}^\times$ is a function
\be
\chi(x) \coloneqq e^{2\pi i \left\{x \right\}}
\ee
such that $\chi(x+y)=\chi(x)\chi(y)$, with the fractional part of $x$ defined as
\be
\{x\} \coloneqq \sum_{i=i_0}^{-1} x_i p^i, \quad \mrm{where} \quad x = \sum_{i=i_0}^\infty x_i p^i \quad and \quad x_i= 0,\dots,p-1.
\ee
The Fourier transform $\FF$ of a function $\psi:\Qp\to\mathbb{C}$ is given by integrating against the additive character,
\be
\FF\psi(k) = \int_\Qp \psi(x) e^{2\pi i \{kx\}}dx.
\ee 
\end{defn}

\begin{defn}
A multiplicative character $\pi_{s,\tau}:\Qp^\times\to \mathbb{C}^\times$ is a function
\be
\pi_{s,\tau}(x)\coloneqq |x|^s\sgn_\tau x,
\ee 
with $s\in\mathbb{C}$ and $\tau\in\Qp$, such that $\pi_{s,\tau}(x_1x_2)=\pi_{s,\tau}(x_1)\pi_{s,\tau}(x_2)$. Here $\sgn_\tau x$ is the $p$-adic sign function (Hilbert symbol). We will not present the sign function here, as it lies outside the scope of this paper; see for instance \cite{Stoica:2018zmi} and references therein.
\end{defn}

We now introduce the Vladimirov derivative. Since our main interest is in the $p$-adic particle-in-a-box, we will give here only the minimum number of facts required in the rest of the paper. For more information on the Vladimirov derivative see Appendix~\ref{appVladDer}. 

\begin{defn}
The Vladimirov derivative associated to character $\pi$ acting on a test function $\psi$ is
\be
\label{eqDder}
\pd_\pi \psi \coloneqq  \FF^{-1} \pi \FF\psi.
\ee
\end{defn}

\begin{comm} Equation \eqref{eqDder} can be written explicitly as
\be
\label{hereispider}
\pd_\pi \psi = \int \pi(k)\psi(x') e^{2\pi i {k(x'-x)}} dx'dk.
\ee
For $\pi=\pi_{s,\tau}$, this expression has a position space representation as
\be
\label{eqDstau}
\pd_{\pi_{s,\tau}} \psi(x) = \Gamma\lb \pi_{s+1,\tau} \rb \int \frac{\psi(x')-\psi(x)}{|x'-x|^{s+1}}\sgn_\tau \lb x'-x \rb dx'.
\ee
When there is no chance of confusion, we will also denote $\pd^{s,\tau}\coloneqq \pd_{\pi_{s,\tau}}$.
\end{comm}

A rigorous derivation of Eq. \eqref{eqDstau} is presented in \cite{HSYZ}, in the sense of distributions. Depending on the regularization, the second term in Eq. \eqref{eqDstau} that is proportional to $\psi(x)$ can be understood to vanish (see Appendix \ref{appVladDer}).

\begin{lemma}
When $\pi_{s,\tau}\neq \pi_{0,1},\ \pi_{-1,1}$, the Vladimirov derivative acting on an additive character gives
\be
\label{eqimm1}
\pd^{s,\tau}_x \chi(kx) = \sgn_\tau(-1) \pi_{s,\tau}(k) \chi(kx).
\ee
\end{lemma}
\begin{proof}
Start from the position space representation \eqref{eqDstau}, perform a shift of the integration variable and use the Fourier transform formula for multiplicative characters to~obtain
\ba
\pd^{s,\tau}_x \chi(kx) &=& \Gamma\lb \pi_{s+1,\tau} \rb \int \pi_{-s-1,\tau} (x') \chi\lsb k(x'+x) \rsb dx' \\
 &=& \Gamma\lb \pi_{s+1,\tau} \rb \frac{\Gamma\lb\pi_{-s,\tau}\rb}{\pi_{-s,\tau}(k)} \chi\lb k x\rb.
\ea
The result follows from the functional equation for the Gamma function.
\end{proof}

\subsection{The spectrum at the finite places}

The first task is to write down a $p$-adic Schr\"odinger equation. In \cite{Stoica:2018zmi}  the following expression was proposed,
\be
\label{hereisschrodi}
H \psi(x,t) = \sgn_{\tau}(-1) |h| \pd^{1,\tau}_{t} \ \psi(x,t).
\ee
Here $\psi:\mathbb{Q}_p^2\to\mathbb{C}$ is the wavefunction, and $H$ is the Hamiltonian. Expression \eqref{hereisschrodi} implies that time and space are now $p$-adic, and in \cite{Stoica:2018zmi} it was argued that the sign factor $\sgn_{\tau}(-1)$ is the correct phase factor in order to ensure that the usual Schr\"odinger equation
\be
\label{hereisschrodiarchi}
\frac{1}{2m} \pd^{2}_{x} \ \psi = -\frac{2 \pi i}{h} \ \pd_{t} \ \psi 
\ee 
can be reconstructed from Eq. \eqref{hereisschrodi}, at the Archimedean place, for the free~particle, using the sign convention in Eq. \eqref{Hsignconv} below. Note that constant $h$ here is the un-reduced Planck's constant. 

$p$-adic time evolution is not the same as Archimedean time evolution, and we will not attempt to discuss it here. Rather, we can focus on the \emph{statics}, and investigate the spectrum of the Hamiltonian. The eigenfunction equation for the free particle Hamiltonian is just
\be
H \psi(x) = E \psi(x),
\ee
where we can take the kinetic term to be
\be
\label{Hsignconv}
H = \left|\frac{h^2}{2m} \right| \pd_x^2.
\ee
Here $\pd_x^2\coloneqq\pd_x^{2,\tau=1}$ is a Vladimirov derivative. As explained in \cite{Stoica:2018zmi}, the $p$-adic norms on unitful quantities should be understood in the sense that when plugged in the Schr\"odinger equation, combinations of unitful quantities can always be constructed such that the norms are applied only to unitless objects. In the case of quantum mechanics, this is the same as pretending that all quantities are unitless. 

We now explain what we mean by particle in a box. In the usual Archimedean quantum mechanics, there are two ways to think about the particle being in a box: (i) an infinite potential outside a certain spatial region, or (ii) periodic boundary conditions. A version of the $p$-adic infinite potential well, which does not appear to be immediately related to the discussion in the present paper, has been considered in \cite{leroy}. In this paper we will only consider option (ii), and impose periodic boundary conditions on our wavefunction. If the box is of size $T\in\Qp$, this means that we must demand
\be
\label{condperiod}
\psi(x+2T) = \psi(x),
\ee
where $\psi:\Qp\to\mathbb{C}$ is the wavefunction. This is because in the Archimedean case the wavefunctions are proportional to $\sin \lb n\pi x/T\rb$ for a box of size $T$, so they repeat after period $2T$.

The wavefunction $\psi(x)$ of the particle can be expanded in Fourier modes as
\be
\psi(x) = \int c_{k} \chi(kx) dk,
\ee
where $c_{k}$ are the Fourier coefficients and $k\in\Qp$. The periodicity condition \eqref{condperiod} enforces that only certain modes are allowed. This is encapsulated in Lemma \ref{lemmafmodes} below.

\begin{lemma}
\label{lemmafmodes}
Periodicity condition \eqref{condperiod} implies that the Fourier coefficients obey
\be
\label{immeq2}
c_k = \begin{cases}
0 \hspace{1.815cm} \mrm{if} \quad 2kT \notin \Zp\\
\mrm{arbitrary} \quad \mrm{if} \quad 2kT \in \Zp
\end{cases}.
\ee
\end{lemma}
\begin{proof}
The condition $\psi(x+2T) = \psi(x)$ means that
\be
\int c_k \chi \lsb k \lb x+2T \rb\rsb dk = \int c_k \chi \lb k x \rb dk.
\ee
Multiplying this relation by $\chi\lb k' x \rb$, integrating in $x$ and using the orthogonality relation $\int \chi\lsb \lb k-k' \rb x\rsb dx=\delta\lb k-k' \rb$, we obtain
\be
c_{k'} \ e^{2\pi i \left\{2k'T \right\}} = c_{k'}.
\ee
So either $e^{2\pi i \left\{2k'T \right\}}=1$, or $c_{k'}=0$. The fractional part of a $p$-adic number is zero iff that number is a $p$-adic integer, and the result follows.
\end{proof}

Finally, we are in a position to characterize the spectrum of the $p$-adic particle-in-a-box at a finite place. 

\begin{lemma}
The energy eigenfunctions of the particle in a box of size $T$ (period $2T$) are $\chi\lb kx \rb$, and the spectrum is given by
\be
\label{hereisspectrum}
E(k) = \left| \frac{ h^2 k^2}{2m}\right|,
\ee
with $k\in\Qp$ such that $2kT\in\Zp$.
\end{lemma}
\begin{proof}
Immediate from Equations \eqref{eqimm1}, \eqref{immeq2}, and definition \eqref{Hsignconv}.
\end{proof}

Note that in order to obtain the spectrum it was not necessary to impose any boundary conditions on the wavefunction, only periodicity. 

\pagebreak
\subsection{Product formula for the spectrum}

Remarkably, Eq. \eqref{hereisspectrum} can be used to reconstruct the spectrum of the Archimedean particle-in-a-box in one dimension, in the sense of \cite{Stoica:2018zmi}. Remembering that Eq.~\eqref{hereisspectrum} is at \emph{one} finite place, but with all places in mind, we restrict to $h,k,m\in\mathbb{Q}_{+}$, and we rewrite Eq. \eqref{hereisspectrum} as
\be
E_{(p)}(k) = \left| \frac{ h^2 k^2}{2m}\right|_{(p)}
\ee
where the dependence on place has now been made explicit. The condition $2kT~\in~\Zp$ can be recast as
\be
|2kT|_{(p)} = p^{q_{(p)}}, 
\ee
with $q_{(p)}\leq 0$ an integer. Then the energy at the finite place is just
\be
E_{(p)} = p^{2q_{(p)}} \left| \frac{ h^2 }{8mT^2}\right|_{(p)}.
\ee
If we define an Archimedean quantity by the Euler product as
\be
\label{Eprodformula}
E_{(\infty)} \coloneqq \prod_p E_{(p)}^{-1},
\ee
we thus obtain
\be
E_{(\infty)} =\frac{ h^2 }{8mT^2} \prod_p p^{-2q_{(p)}}.
\ee
Remembering that each $q_{(p)}$ obeys $q_{(p)}\leq 0$ and assuming that only a finite number of primes enter the product with nontrivial exponents, this is the same as
\be
E_{(\infty)} =\frac{ n^2 h^2 }{8mT^2},
\ee
with $n$ a positive integer. We have recovered the spectrum of the one dimensional quantum particle-in-a-box.

\section{Archimedean from $p$-adic: Flow in the Berkovich space~$\MM\lb \mathbb{Z}\rb$}
\label{secBerk}

In this section we will propose an interpretation for the product formula \eqref{Eprodformula}, in terms of a mathematical object called \emph{Berkovich space}.

A priori, it may seem strange that the spectrum of the particle-in-a-box should obey an adelic product. In fact, in the literature there are currently two known ways, both empirical, to relate $p$-adic and Archimedean objects: (i) adelic products and (ii) the $p\to1$ limit. In this section we would like to propose Berkovich spaces as the mathematical arena in which both of these approaches can be systematically~understood.

In mathematics, Berkovich spaces can be thought of as spaces of seminorms, both $p$-adic and Archimedean. We will not give a mathematical overview of these spaces here, see instead \cite{baker,jonsson}.  Rather, we will focus on the physical aspects. We propose that from a physicist's point of view, Berkovich spaces can be thought of as spaces of theories, together with the domains on which the theories are defined. The Berkovich spaces thus encode a flow of theories. We furthermore propose that this flow should be interpreted as a kind of renormalization group flow.

\subsection{A physicist's review: The Berkovich space $\MM(\mathbb{Z})$}

In this section we will only be concerned with $\MM\lb \mathbb{Z} \rb$, which is a rigid analytic space associated to the normed ring $(\mathbb{Z},|\cdot|_\infty)$, with $|\cdot|_\infty$ the usual absolute value norm. We leave the exploration of other Berkovich spaces in the context of physics for future~work. We will be schematic in our presentation; for more details see e.g. the notes by Baker \cite{baker}, which we will mostly follow, or \cite{jonsson}.

\begin{figure}[htp]
\centering
\begin{tikzpicture}
\tikzset{VertexStyle/.style = simple}
\Vertex[L=$ $,x=0,y=0,]{u0}
\Vertex[L=$ $,x=0,y=3.5]{u1}
\Vertex[L=$ $,x=0,y=1.75]{u2}
\Vertex[L=$ $,x=-3,y=-3]{u3}
\Vertex[L=$ $,x=4.4,y=-3]{u4}
\Vertex[L=$ $,x=-0.6,y=-3]{u5}
\Vertex[L=$ $,x=-1.5,y=-1.5]{u6}
\Vertex[L=$ $,x=-0.3,y=-1.5]{u7}
\Vertex[L=$ $,x=2.2,y=-1.5]{u8}
\Vertex[L=$ $,x=1.1,y=-0.75]{u8p}
\Edge[](u0)(u1)
\Edge[](u0)(u2)
\Edge[](u0)(u3)
\Edge[](u0)(u4)
\Edge[](u0)(u5)
\tikzset{VertexStyle/.style = none}
\Vertex[L=$|\cdot|_3$,x=0.35,y=-1.5]{l0}
\Vertex[L=$|\cdot|_2$,x=-2.6+0.45,y=-1.5]{l1}
\Vertex[L=$|\cdot|_p$,x=3.35-0.45,y=-1.5]{l2}
\Vertex[L=$|\cdot|^\epsilon_p$,x=2.3-0.45,y=-0.75]{l2p}
\Vertex[L=$|\cdot|_{3,\infty}$,x=0.7-0.45,y=-3]{l3}
\Vertex[L=$|\cdot|_{2,\infty}$,x=-4.3+0.5,y=-3]{l4}
\Vertex[L=$|\cdot|_{p,\infty}$,x=5.7-0.45,y=-3]{l5}
\Vertex[L=$\dots$,x=2.8,y=-3]{l6}
\Vertex[L=$|\cdot|_0$,x=-0.65,y=0]{l7}
\Vertex[L=$|\cdot|_{\infty}^\epsilon$,x=1.1-0.4,y=1.75]{l8}
\Vertex[L=$|\cdot|_{\infty}$,x=1.1-0.4,y=3.5]{l9}
\Vertex[L=$\epsilon$,x=0.4,y=0.875]{l8e}
\Vertex[L=$\epsilon$,x=0.75,y=-0.1]{l8f}
\Vertex[L=$ $,x=0.18,y=1.75]{l10}
\Vertex[L=$ $,x=0.18,y=0]{l11}
\Vertex[L=$ $,x=0.23,y=+0.02]{l11p}
\Vertex[L=$ $,x=1.17,y=-0.62]{l12}
\Edge[style=->](l11)(l10)
\Edge[style=->](l11p)(l12)
\end{tikzpicture}
\caption{The Berkovich space $\MM\lb\mathbb{Z}\rb$ is a space of seminorms, applied to elements $x\in\mathbb{Z}$.}
\label{fig1}
\end{figure}
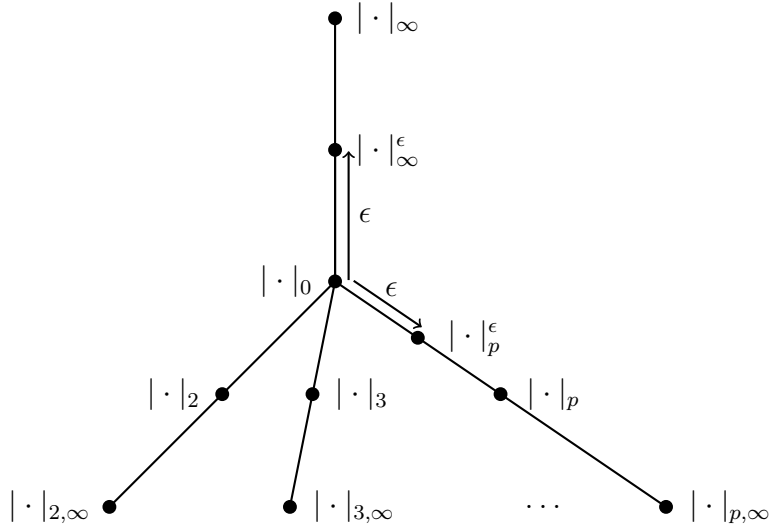

The Berkovich space $\MM(\mathbb{Z})$ is represented in Figure \ref{fig1}. It has the structure of a vertex $|\cdot|_0$, out of which an infinite number of branches emanate. One branch corresponds to the Archimedean direction, and the other branches correspond to the $p$-adic directions; each point in $\MM\lb \mathbb{Z} \rb$ is a seminorm. 

We will now explain this construction.

\begin{defn}
Let $\mathbb{K}$ be a ring. A seminorm $|\cdot|$ is a function $|\cdot|:\mathbb{K}\to\mathbb{R}_{\geq0}$ with the following two~properties:
\begin{enumerate}
\item $|x+y|\leq |x|+|y|$.
\item $|xy|=|x||y|$.
\end{enumerate}
\end{defn}

\begin{comm}
Any norm is also a seminorm. A norm obeys the extra property that if $|x|=0$ then $x=0$.
\end{comm}

The seminorms entering the space $\MM(\mathbb{Z})$ are the following (here $x\in \mathbb{Z}$):
\begin{enumerate}
\item The trivial norm $|\cdot|_0$, defined as
\be
|x|_0\coloneqq \begin{cases}
0  \quad \mrm{if}\quad  x = 0 \\
1 \quad \mrm{if}\quad x\neq 0
\end{cases}.
\ee
\item The Archimedean norms $|\cdot|_\infty^\epsilon$, with $0<\epsilon\leq1$.
\item The $p$-adic norms $|\cdot|_p^\epsilon$, with $0<\epsilon<\infty$.
\item The $p$-trivial seminorms  $|\cdot|_{p,\infty}$, defined as
\be
|x|_{p,\infty} \coloneqq \begin{cases}
0 \quad \mrm{if} \quad p\hspace{1.15mm} |\,\hspace{0.5mm} n \\
1 \quad \mrm{if} \quad p\nmid n
\end{cases}.
\ee
\end{enumerate}

\begin{lemma}
\label{lemma4}
For $x\in\mathbb{Z}$, the seminorms above obey the limits
\begin{enumerate}
\item $\lim_{\epsilon\to0}|x|^\epsilon_\infty=|x|_0$.
\item $\lim_{\epsilon\to0}|x|^\epsilon_p=|x|_0$.
\item $\lim_{\epsilon\to\infty}|x|^\epsilon_p=|x|_{p,\infty}$.
\end{enumerate}
\end{lemma}
\begin{proof}
Immediate.
\end{proof}

The structure of the Berkovich space $\MM(\mathbb{Z})$ represented in Figure \ref{fig1} arises from Lemma \ref{lemma4}.

\begin{comm}
Lemma \ref{lemma4} does not hold if we demand $x\in\mathbb{Q}$. In this case, one can instead consider the Berkovich projective line (see \cite{baker}). We will remark more on the physical interpretation of the Berkovich projective line in Section \ref{sec32} below.
\end{comm}

\subsection{$\MM(\mathbb{Z})$ as renormalization group~flow}
\label{sec32}

We now interpret the Berkovich space $\MM\lb\mathbb{Z}\rb$ as renormalization group flow. In this subsection only we will take $h/(2m)=2T=1$, so  that the energy at a place $v$ ($p$-adic or Archimedean) is
\be
E_{(v)}= |k|^2_{(v)}.
\ee
Since we are interested in values of $k$ valid for all places, and remembering the periodicity condition $|2kT|_{(p)}=1$, we can take $k\in\mathbb{Z}$, so that $\MM\lb\mathbb{Z}\rb$ is the correct object to consider. Note however that restricting to integers, although sufficient for the purpose of the present paper, is likely unnecessarily  restrictive when considering flows associated to Euler products such as in Eq. \eqref{Eprodformula}, and can probably be relaxed. We will come back to this point in Section \ref{secdisc} below.

\begin{figure}[htp]
\centering
\begin{tikzpicture}
\tikzset{VertexStyle/.style = simple}
\Vertex[L=$ $,x=0,y=0,]{u0}
\Vertex[L=$ $,x=0,y=3.5]{u1}
\Vertex[L=$ $,x=0,y=1.75]{u2}
\Vertex[L=$ $,x=-3,y=-3]{u3}
\Vertex[L=$ $,x=4.4,y=-3]{u4}
\Vertex[L=$ $,x=-0.6,y=-3]{u5}
\Vertex[L=$ $,x=-1.5,y=-1.5]{u6}
\Vertex[L=$ $,x=-0.3,y=-1.5]{u7}
\Vertex[L=$ $,x=2.2,y=-1.5]{u8}
\Vertex[L=$ $,x=1.1,y=-0.75]{u8p}
\Edge[](u0)(u1)
\Edge[](u0)(u2)
\Edge[](u0)(u3)
\Edge[](u0)(u4)
\Edge[](u0)(u5)
\tikzset{VertexStyle/.style = none}
\Vertex[L=$\lb\mathbb{Q}_3{,}\,|\cdot|_3\rb$,x=0.8,y=-1.5]{l0}
\Vertex[L=$\lb\mathbb{Q}_2{,}\,|\cdot|_2\rb$,x=-2.6,y=-1.5]{l1}
\Vertex[L=$\lb\mathbb{Q}_p{,}\,|\cdot|_p\rb$,x=3.35,y=-1.5]{l2}
\Vertex[L=$\lb\mathbb{Q}_p{,}\,|\cdot|^\epsilon_p\rb$,x=2.3,y=-0.75]{l2p}
\Vertex[L=$\lb\mathbb{Q}_3{,}\,|\cdot|_{3,\infty}\rb$,x=0.7,y=-3]{l3}
\Vertex[L=$\lb\mathbb{Q}_2{,}\,|\cdot|_{2,\infty}\rb$,x=-4.3,y=-3]{l4}
\Vertex[L=$\lb\mathbb{Q}_p{,}\,|\cdot|_{p,\infty}\rb$,x=5.7,y=-3]{l5}
\Vertex[L=$\dots$,x=2.8,y=-3]{l6}
\Vertex[L=$|\cdot|_0$,x=-0.65,y=0]{l7}
\Vertex[L=$\lb\mathbb{R}{,}\,|\cdot|_{\infty}^\epsilon\rb$,x=1.1,y=1.75]{l8}
\Vertex[L=$\lb\mathbb{R}{,}\,|\cdot|_{\infty}\rb$,x=1.1,y=3.5]{l9}
\Vertex[L=$\epsilon$,x=0.4,y=0.875]{l8e}
\Vertex[L=$\epsilon$,x=0.75,y=-0.1]{l8f}
\Vertex[L=$ $,x=0.18,y=1.75]{l10}
\Vertex[L=$ $,x=0.18,y=0]{l11}
\Vertex[L=$ $,x=0.23,y=+0.02]{l11p}
\Vertex[L=$ $,x=1.17,y=-0.62]{l12}
\Edge[style=->](l11)(l10)
\Edge[style=->](l11p)(l12)
\end{tikzpicture}
\caption{The Berkovich space as a space of theories.}
\label{fig2}
\end{figure}
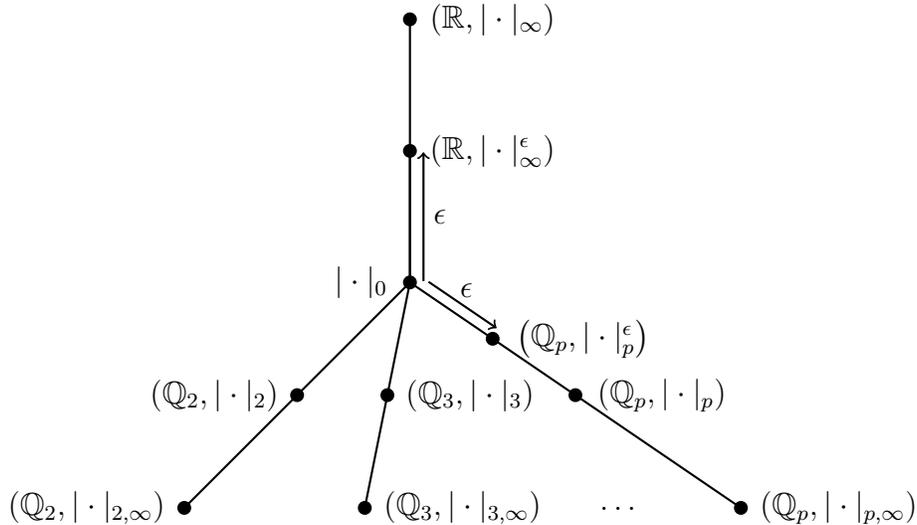

The physicists' interpretation of Berkovich space is in Figure \ref{fig2}. In this interpretation we \emph{augment} the Berkovich space from a space of norms to a space of \emph{theories}. That is, at each point we have a $p$-adic or Archimedean theory, defined on $\mathbb{R}$ or $\mathbb{Q_p}$ equipped with a seminorm. The Berkovich space thus encodes a flow of theories.

\subsubsection{Norms as energy scales}

Consider a point on a $p$-adic or Archimedean branch $v$ of $\MM(\mathbb{Z})$, parameterized by $\epsilon$, as in Figure \ref{fig2}. Since moving along the branch corresponds to scaling the norm, we associate an energy to this point as 
\be
E_{(v)}(\epsilon)= |k|_{(v)}^{2\epsilon}.
\ee
Redoing the analysis of Section \ref{freeparticle}, this energy is an eigenenergy of the particle-in-a-box, with all norms appropriately scaled by $\epsilon$. This motivates interpreting the Berkovich space as a space of theories. In the rest of this section we will focus mostly on the flow of the energy $E_{(v)}(\epsilon)$. Understanding directly how the theories flow in Berkovich space is important, and will have many powerful applications, however we will not consider this question more in the rest of the paper.

Now consider the Archimedean branch $v=\infty$. As we move away from the point $|\cdot|_{(\infty)}$, the parameter $\epsilon$ decreases. We can think of this as a passive transformation of the energy scale, and the corresponding active transformation is increasing the value of $|k|_{(\infty)}^2$. We thus interpret moving along the $v=\infty$ branch of $\MM\lb\mathbb{Z}\rb$ from $|\cdot|_{(\infty)}$ toward $|\cdot|_{(0)}$ as moving toward the ultraviolet. Flowing along the branch $v=\infty$ is therefore a kind of renormalization group flow. This flow continues until we reach $|\cdot|_{(0)}$, at which point the flow splits into separate branches, one for each prime. These branches, and their Lagrangians, can be thought of as providing a $p$-adic ultraviolet completion to the particle-in-a-box.

Of course, the Archimedean particle in a box is a free theory in quantum mechanics, and as such the physics doesn't change with the energy scale. In the present language this simply means that the theory scales trivially along the $v=\infty$ branch, with no interesting dynamics. Mathematically (i.e. assuming that periodic boundary conditions can be imposed at any energy scale), the theory is well-defined at any energy. Nonetheless, the $p$-adic branches can still be used to provide a $p$-adic ultraviolet completion, even for a free theory.

We should also remark on the endpoints of the flow along the $p$-adic branches. One may naively assume the the flow ends at the $p$-adic theories, i.e. the points $|\cdot|^{\epsilon=1}_{(p)}$, however this is not correct. Rather, the flow continues all the way to the trivial seminorms $|\cdot|_{p,\infty}$. We will not attempt to define theories on $\lb \Qp, |\cdot|_{p,\infty} \rb$ in this paper, but it would be interesting to do so.

\subsubsection{Flow equation for the energy: first order analysis}

Let's now write down a flow equation for the renormalization group flow for the energy eigenvalues $E_{(v)}(\epsilon)$, on the Berkovich space. Since the particle-in-a-box is a free theory, the most natural object to consider is an equation of motion involving the Laplacian on $\MM\lb\mathbb{Z}\rb$. We will not attempt to derive this equation from first principles here (and indeed the question of whether this equation of motion should be derived, or is more akin to a definition of how the Archimedean world emerges from the $p$-adic ones, will be left for future work).

Let's consider the equation
\be
\label{lap1}
\Delta_{(1)} E(x) = 0,
\ee
at a point $x\in\MM(\mathbb{Z})$, where $\Delta$ is the Laplacian on the Berkovich space, and the subscript controls the order of expansion in the Laplacian, in a way that will be made precise immediately below. The Laplacian involves a sum over the branches neighboring point $x$, so that if $x= |\cdot|_{(v)}^{\epsilon_0}$, then Eq. \eqref{lap1} is shorthand for
\be 
\label{eqq36}
\lim_{\epsilon\to0} \frac{E_{(v)}(\epsilon_0+\epsilon)+E_{(v)}(\epsilon_0-\epsilon)-2E_{(v)}(\epsilon_0)}{\epsilon} = 0,
\ee
and at the center $x=|\cdot|_0$ Eq. \eqref{lap1} means (the sum is over all places, i.e. branches)
\be
\label{eqq37}
\lim_{\epsilon\to0} \sum_{v} \frac{E_{(v)}\lb\epsilon\rb - E\lb|\cdot|_0\rb}{\epsilon} = 0.
\ee
The subscript $(1)$ on the Laplacian (and corresponding factor of $\epsilon$, not $\epsilon^2$, in the denominator) means that, for now, we will be expanding the energies to first order in~$\epsilon$. 

Expanding $E_{(v)}$ in a power series around $\epsilon_0$ on any of the branches, it is immediate that Equation \eqref{eqq36} is trivially satisfied by the solution
\be
\label{eqsol38}
E_{(v)}(\epsilon) = |k|_{(v)}^{2\epsilon}. 
\ee
The fact that Eq. \eqref{eqq36} does not constrain solution \eqref{eqsol38}, to this order in $\epsilon$, can be seen as another restatement of the fact that the evolution along the branches is trivial (i.e. pure scaling). Note also that because we are for now truncating the order in the Laplacian, Eq. \eqref{eqq36} does not fully fix the solution \eqref{eqsol38}, since all the higher $\epsilon$ terms in the expansion of $E_{(v)}(\epsilon)$ get dropped; we will remedy this below.

Consider now the expansion of Eq. \eqref{eqq37} around the center $|\cdot|_0$. Dropping the higher terms as $\epsilon\to0$, and plugging in solution \eqref{eqsol38}, Eq. \eqref{eqq37} becomes
\be
\sum_v \ln |k|_{(v)}=0,
\ee
i.e. we recover the product
\be
\prod_v |k|_{(v)}=1.
\ee
Thus, the content of Eq. \eqref{lap1} at the point $|\cdot|_{(0)}$ is the same as the Euler product for the norms which allows reconstructing the Archimedean energy according to Eq.~\eqref{Eprodformula}.

\subsubsection{The full flow equation}
The argument at leading order in $\epsilon$ in the subsection above captures the essence of how Eq. \eqref{lap1} encodes the product formula.  Let's now go beyond leading order in~$\epsilon$, and trace the flow from the theories defined on $\lb \Qp,|\cdot|_{(p)} \rb$ to the theory at $\lb \mathbb{R},|\cdot|_{(\infty)} \rb$.

We write down the Berkovich flow equation as
\be
\Delta E(x) = \lambda^2 E(x),
\ee
where $\Delta$ is now the Laplacian to second order in $\epsilon$, and $\lambda$ a parameter that will be introduced shortly. Spelling out this equation at a point $|\cdot|^{\epsilon_0}_{(v)}$ on a branch $(v)$ we have
\be
\label{eq312}
\lim_{\epsilon\to0} \frac{E_{(v)}(\epsilon_0+\epsilon)+E_{(v)}(\epsilon_0-\epsilon)-2E_{(v)}(\epsilon_0)}{\epsilon^2} = \lambda_{(v)}^2 E_{(v)}\lb \epsilon_0 \rb,
\ee
and at point $|\cdot|_0$ we have
\be
\lim_{\epsilon\to0} \sum_{v} \frac{E_{(v)}\lb\epsilon\rb - E\lb|\cdot|_0\rb}{\epsilon^2} = \lambda_{|\cdot|_0}^2 E\lb|\cdot|_0\rb.
\ee
For the values of the parameters $\lambda$ we pick
\be
\lambda_{(v)} = \ln |k|^2_{(v)}
\ee
on a branch $(v)$, and 
\be
\lambda^2_{|\cdot|_{(0)}} = \frac{1}{2} \sum_v \lb \ln |k|_{(v)}^2 \rb^2
\ee
at the center point.

Expanding in $\epsilon$ and then renaming $\epsilon_0$ to $\epsilon$, equation \eqref{eq312} is simply
\be
\label{eq316}
\pd^2_\epsilon E(\epsilon) = \lambda^2_{(v)} E(\epsilon).
\ee
Consider now the flow along a $p$-adic branch. Imposing the boundary condition at $\lb\Qp,|\cdot|\rb$ that $E_{(p)}(\epsilon)=|k|_{(p)}^2$, and demanding that the energy does not diverge toward the $p$-adic ultraviolet $\lb \Qp,|\cdot|_{p,\infty}\rb$ as $\epsilon\to\infty$, the flow along the $p$-adic branches is now uniquely determined by Eq. \eqref{eq316} as
\be
E_{(p)}(\epsilon) = |k|_{(p)}^{2\epsilon}.
\ee
Solving the rest of the flow equations at $|\cdot|_0$ and on the Archimedean branch, and demanding that the flow should increase monotonically toward $(\mathbb{R},|\cdot|_{(\infty)})$, the energy gets specified everywhere on the tree by the flow equation. Thus, the flow together with the values of the energy at the $p$-adic points $\lb \Qp,|\cdot|_{(p)}\rb$ are enough to determine the value of the Archimedean energy.

\begin{comm}
For the space $\MM\lb\mathbb{Z}\rb$ (as for graphs), the distinction between first order and second order differential operators is not sharp. Thus, one could also consider other flow equations, based on first order operators, which are quantitatively similar to the equations discussed above. For instance, consider the flow equations 
\ba
\lim_{\substack{\epsilon\to0\\\epsilon>0}} \frac{E_{(v)}(\epsilon_0+\epsilon)-E_{(v)}(\epsilon_0)}{\epsilon} &=& \lambda_{(v)} E_{(v)}\lb \epsilon_0 \rb,\\
\lim_{\epsilon\to0} \sum_{v} \frac{E_{(v)}\lb\epsilon\rb - E\lb|\cdot|_0\rb}{\epsilon} &=& 0,
\ea
on the branches and at $|\cdot|_0$ respectively. These equations also uniquely determine the flow \eqref{eqsol38}, and thus the Archimedean eigenenergy, given the $p$-adic eigenenergies as boundary conditions. Furthermore, using first order equations does not need assumptions on the direction of increase of the flow along the branches.
\end{comm}

\begin{comm}
The form of the flow equation for the eigenenergies is strongly dependent on the fact that this is a free theory in quantum mechanics. In fact, the Berkovich equation of motion is the same as the free particle eigenequation, however it is unclear if this has a deeper meaning. For more complicated quantities in field theory that flow nontrivially with the energy scale, one should of course expect other Berkovich equations of motion.
\end{comm}

\subsection{Toward other renormalization group flows}
\label{secdisc}

We will end this section with a few comments. First, we should note that as explained in the introduction above, there are many physical quantities which obey Euler product formulas, such as the four-point Virasoro amplitudes for open strings at tree level. The  equations of flow in Berkovich space discussed above should apply to all such objects. Furthermore, the Berkovich space provides a natural setting for reconstructing (in the sense of \cite{Stoica:2018zmi}) more complicated Archimedean objects which don't satisfy product rules, such as five-point and higher point amplitudes. This is because theories along the branches of the Berkovich deform continuously from one to the other, so any physical quantities in these theories will also deform continuously. Thus, any physical quantity will obey an equation of motion in Berkovich space, however for more complicated objects these equations of motion will not lead to simple product formulas. We leave the study of such more general Berkovich equations of motion to future work (though we will discuss below the $p\to1$ procedure).

\begin{figure}[htp]
\centering
\includegraphics[scale=0.5]{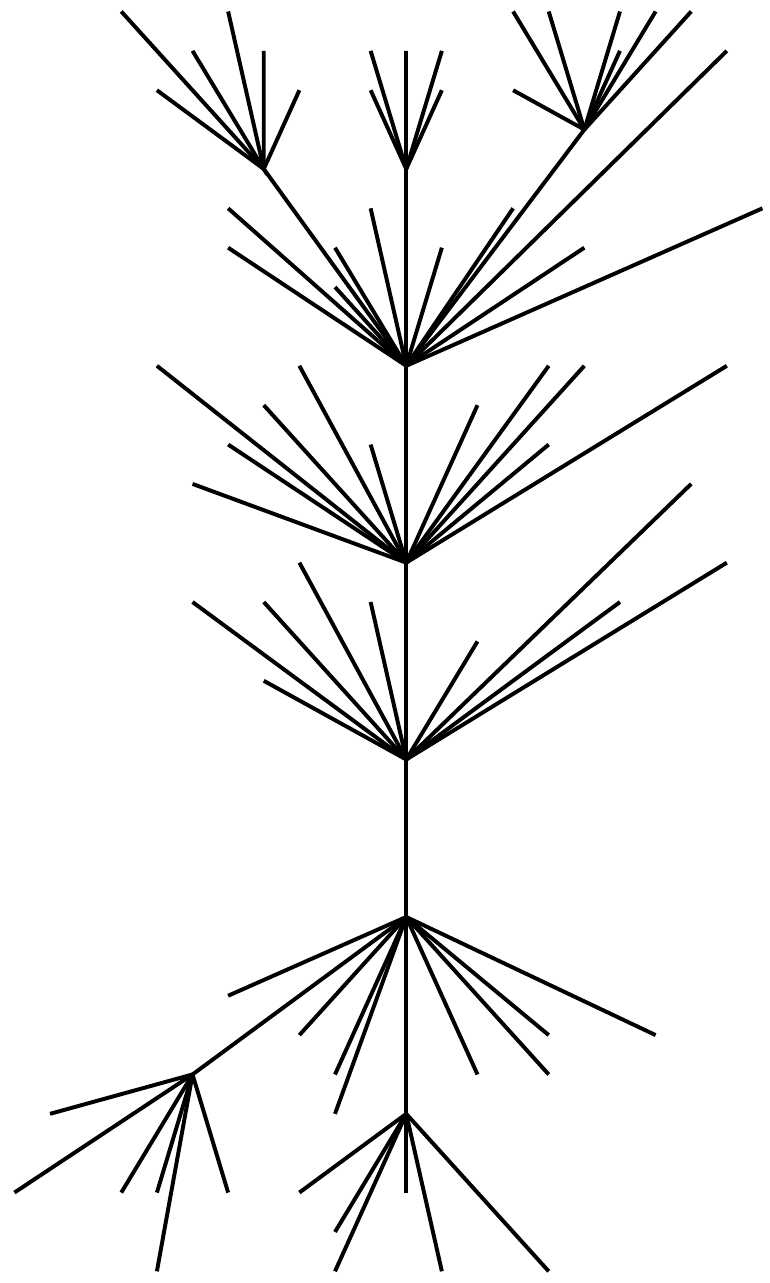}
\caption{A schematic representation of the Berkovich projective line $\mathbb{P}_\mrm{Berk}^1$.}
\label{fig3}
\end{figure}

Let's now comment on the integer restriction. In the analysis above the energies obey $E\in\mathbb{Z}$. If instead $E\in\mathbb{Q}-\mathbb{Z}$, then the discussion becomes plagued by ultraviolet divergences, as $E$ will diverge when $\epsilon\to\infty$ on the $p$-adic branches. This signals that the $p$-adic ultraviolet completion breaks down. However, the restriction to integers is in some sense unphysical, as it can easily be relaxed, for instance by restoring $m$, $h$ and $T$. Thus, there must exist a generalization of our discussion that can incorporate energies (and other quantities) that are $\mathbb{Q}$-valued. This generalization is given by considering instead the Berkovich projective line $\mathbb{P}_\mrm{Berk}^1$. We will not review the Berkovich projective line here, see for instance \cite{baker,jonsson}, however we will give a brief description. The Berkovich projective line $\mathbb{P}^1_\mrm{Berk}$ can be thought of as a central vertex, out of which an infinite number of branches emanate (see Figure \ref{fig3}). However, unlike in the case $\MM\lb\mathcal{Z}\rb$, infinitely many branches can now emanate from certain points on the branches also. The Berkovich projective line is colloquially known as a witch's broom.

We now relate the Berkovich space with the $p\to1$ limit. As we have mentioned above, the $p\to1$ limit can be used to obtain an Archimedean tachyonic Lagrangian with logarithmic potential from an effective $p$-adic Lagrangian \cite{gershata}. Recently, in \cite{Bocardo-Gaspar:2017atv,bggczg} it was established that the four-point and five-point scattering amplitudes for this Archimedean Lagrangian follow from the corresponding $p$-adic amplitudes, by taking the $p\to1$ limit in a controlled manner. The first step in the procedure of \cite{Bocardo-Gaspar:2017atv,bggczg} is to pass from $\Qp$ to an unramified field extension, which can be understood as moving in the Berkovich projective line. The limit $e\to0$ can also be given the interpretation of Berkovich flow, thus we propose that an interpretation of the $p\to1$ procedure is that it is another way of keeping track how physical quantities change along different directions in Berkovich space. If this interpretation is correct, the two ways of reconstructing Archimedean quantities (Euler products and the $p \to1$ limit) can both be understood in the Berkovich space framework. We defer a more detailed analysis to future work. 

\section{Discussion}
\label{sec4}

Let's now end with a mathematical comment. The philosophy of flowing in Berkovich space described in this paper is reminiscent of how general covariance functions in general relativity, if one identifies different coordinate systems with places. That is, in general relativity one can write down covariant statements that take the same form in all coordinate systems. The reason such statements are coordinate invariant is because they remain unchanged under a coordinate transformation that deforms from one coordinate system to another. This is similar to what we have have been discussing in this paper: one can write down statements (such as the Schrodinger equations \eqref{hereisschrodi}, \eqref{hereisschrodiarchi}, or other formulas such as the integral representations of $n$-point Veneziano amplitudes at tree level) that take the same form at $\mathbb{R}$ and $\Qp$. A priori it may seem surprising that the same expressions should hold both for $\mathbb{R}$ and $\Qp$, however if it is possible to deform from one to the other along paths in Berkovich space, such that the expressions remain invariant (similarly to an infinitesimal coordinate transformation) then that the same expressions apply across places is not surprising, but instead it is \emph{natural}. Of course, the invariance under infinitesimal deformations in Berkovich space will not apply to all quantities, just as not all expressions are covariant in general relativity. It would be interesting to systematically classify which objects obey this.

\section*{Acknowledgments}

We thank H. Hampapura, C. Jepsen, M. M. Nastasescu, Dingxin Zhang, and  W. A. Z\'u\~niga-Galindo for useful discussions. The work of A.H. was supported in part by a grant from the Simons Foundation in Homological Mirror Symmetry. The work of A.H., D.M., and B.S. was supported in part by a grant from the Brandeis University Provost Office. B.S. was supported in part by the U.S. Department of Energy under grant DE-SC-0009987, and by the Simons Foundation through the It from Qubit Simons Collaboration on Quantum Fields, Gravity and Information.

\appendix

\section{Vladimirov derivatives}
\label{appVladDer}

This appendix contains some results on Vladimirov derivatives. These results are written informally, in the style of physicists, and are included for completeness, as they are mostly not needed for the main content of the paper. See \cite{HSYZ} for a rigorous discussion of Vladimirov derivatives with nontrivial characters.

\subsection{Basics}
The standard literature on Vladimirov derivatives (see e.g. \cite{vvzbook,zunigabook}, and \cite{GGIP} for an introduction to the theory of characters on $\Qp$) contains two flavors for the Vladimirov derivative: (i) unregularized, and (ii) regularized. This is because the integrals required for the calculation of the derivative typically diverge when applied to certain classes of functions, such as polynomials, or even constants. To regularize these divergences, the regularized version is sometimes used.

\begin{defn}
Let $s\in \mathbb{R},\ \tau\in \Qp$. The position space Vladimirov derivative $\pd^{s,\tau}$ acting on a function $\psi(x)$ is associated to a multiplicative character $\pi_{s,\tau}$, and is defined as follows:
\begin{enumerate}
\item Take the Fourier transform,
\be
\label{eq21}
\psi(k) = \int \psi(x) \chi(kx).
\ee
\item Apply the multiplicative character,
\be
\label{eq22}
\pd^{s,\tau} \psi(k) = \pi_{s,\tau}(k) \psi(k). 
\ee
\item Fourier transform back,
\be
\label{eq23}
\pd^{s,\tau} \psi(x) = \int \chi(-kx) \pd^{s,\tau} \psi(k).
\ee
\end{enumerate}
\end{defn}

\begin{lemma}
For $\pi_{s,\tau}\neq \pi_{-1,1}$ and $\pi_{s,\tau}\neq \pi_{0,1}$, the Vladimirov derivative has a position space representation given by
\be
\label{eq24}
\pd^{s,\tau} \psi(x) = \Gamma\lb \pi_{s+1,\tau} \rb  \int \frac{ \psi(x') \sgn_\tau\lb x'-x \rb}{|x'-x|^{s+1}} .
\ee
\end{lemma}

\begin{proof}
Perform the $k$ integral, using the integral representation of the Dirac-delta function.
\end{proof}

\begin{comm}
Eq. \eqref{eq24} is the unregularized Vladimirov derivative. A more careful treatment of the integrals, in the sense of distributions, produces an extra term (see \cite{HSYZ}), yielding
\be
\label{vladireg}
D^{s,\tau} \psi(x) = \Gamma\lb \pi_{s+1,\tau} \rb  \int \frac{ \lb\psi(x')-\psi(x)\rb \sgn_\tau\lb x'-x \rb}{|x'-x|^{s+1}} .
\ee
Expression \eqref{vladireg} is the regularized Vladimirov derivative.
\end{comm}

\begin{comm}
The distinction between unregularized and regularized Vladimirov derivatives is not important for the purposes of this paper. Furthermore, adopting Beta function regularization, which will be explained below, the extra term proportional to $\psi(x)$ in the regularized Vladimirov derivative vanishes, and the two versions of derivative coincide.
\end{comm}

\begin{lemma}
\label{lemma2}
For $\pi_{s,\tau}=\pi_{0,1}$ the Vladimirov derivative acts as the identity operator,
\be
\pd^{0,1}_x \psi(x) = \psi(x).
\ee
\end{lemma}
\begin{proof}
Immediate from Eq. \eqref{hereispider} and the integral representation of the Dirac-delta function.
\end{proof}

\begin{lemma}
For all values of parameters $s_i\in\mathbb{R}$, $\tau_i\in\Qp$, the Vladimirov derivative obeys
\be
\label{28}
\pd^{s_1,\tau_1} \pd^{s_2, \tau_2} = \pd^{s_2,\tau_2} \pd^{s_1, \tau_1} = \pd^{s_1+s_2,\tau_1\tau_2}.
\ee
\end{lemma}

\begin{proof}
Use the momentum space representation,
\ba
\pd^{{s_2},\tau_2} \psi(x) &=& \int \pi_{{s_2},\tau_2}(k_2) \chi\lsb k_2\lb x_2-x \rb \rsb \psi(x_2), \\
\pd^{{s_1},\tau_1} \pd^{{s_2},\tau_2} \psi(x) &=& \int \pi_{{s_1},\tau_1}(k_1) \pi_{{s_2},\tau_2}(k_2) \chi\lsb k_1\lb x_1 - x \rb + k_2\lb x_2 - x_1 \rb  \rsb \psi(x_2).
\ea
The $x_1$ integral gives a delta function, which sets $k_1=k_2$, so that
\ba
\pd^{{s_1},\tau_1} \pd^{{s_2},\tau_2} \psi(x) &=& \int \pi_{{s_1+s_2},\tau_1\tau_2}(k) \chi\lsb k\lb x' - x \rb   \rsb \psi(x') \\
&=& \pd^{s_1+s_2,\tau_1\tau_2} \psi(x).
\ea
\noindent For $s_1+s_2=0$, $\tau_1\tau_2=1$, the derivative on the right-hand side in Eq. \eqref{28} is the identity operator. For $s_1+s_2=-1$, $\tau_1\tau_2=1$, the derivative on the right-hand side doesn't have a position space representation, however Eq. \eqref{28} still holds formally.
\end{proof}

\begin{lemma}
For all $s\in\mathbb{R}$, $\tau\in\Qp$, constants can be pulled in front of the Vladimirov derivative,
\be
\pd^{s,\tau}_{x} c \psi (x) = c \pd^{s,\tau}_x \psi(x).
\ee
\end{lemma}
\begin{proof}
If a position representation exists,
\be
\pd^{s,\tau}_{x} c \psi (x) =  c \Gamma\lb \pi_{s+1,\tau} \rb  \int \frac{\psi(x') \sgn_\tau\lb x'-x \rb}{|x'-x|^{s+1}} = c \pd^{s,\tau}_{x} \psi (x).
\ee
If we're in the $s=-1,\tau=1$ case, this still follows from the momentum space representation.
\end{proof}

\begin{fact}
The Vladimirov derivative behaves under translations as
\be
\pd_x^{s,\tau} \psi(x-a) = \pd_{x'}^{s,\tau} \psi(x') \Big|_{x'=x-a}.
\ee
\end{fact}

Let's now introduce the Beta function, which will be needed for discussing regularization below.
\begin{defn}
The Beta function of two multiplicative characters $\pi_{1,2}:\Qp\to\mathbb{C}^\times$ is
\be
\label{B113}
\Beta(\pi_1,\pi_2) = \int \pi_1(x)\pi_2(1-x)|x|^{-1}|1-x|^{-1}.
\ee
\end{defn}
Integral \eqref{B113} can be split into integrals on $\Zp$ and $\Qp$. Just as in the definition of the Gamma function, for generic characters one of the integrals converges, the other does not, and each of the integrals can be analytically continued to a finite answer.

Integral \eqref{B113} evaluates to \cite{GGIP}
\be
\label{B114}
\Beta(\pi_1,\pi_2) = \frac{\Gamma(\pi_1)\Gamma(\pi_2)}{\Gamma(\pi_1\pi_2)},
\ee
so the result can be recast as
\be
\label{B115}
\int \pi_{s_1-1,\tau_1}(x) \pi_{s_2-1,\tau_2}(1-x) = \frac{\Gamma(\pi_{s_1,\tau_1})\Gamma(\pi_{s_2,\tau_2})}{\Gamma(\pi_{s_1+s_2,\tau_1\tau_2})},
\ee
whenever the expression is non-singular.

\subsection{Regularization}

We now discuss the regularized Vladimirov derivative, which in Eq. \eqref{vladireg} was defined~as
\be
D^{s,\tau} \psi(x) \coloneqq \Gamma\lb \pi_{s+1,\tau} \rb  \int \frac{ \psi(x') - \psi(x) }{|x'-x|^{s+1}} \sgn_\tau\lb x'-x \rb.
\ee

\begin{lemma}
The regularized Vladimirov derivative obeys
\be
D^{s_1,\tau_1}D^{s_2,\tau_2} = D^{s_2,\tau_2}D^{s_1,\tau_1}.
\ee
\end{lemma}
\begin{proof}
We can formally write
\be
D^{s,\tau} \psi(x) = \pd^{s,\tau}\psi(x) - \psi(x)\lb \pd^{s,\tau} 1\rb.
\ee
Then
\ba
D^{s_1,\tau_1} D^{s_2,\tau_2} \psi(x) &=& D^{s_1,\tau_1} \lsb \pd^{s_2,\tau_2}\psi(x) - \psi(x)\lb \pd^{s_2,\tau_2} 1\rb \rsb \\
\label{219}
&=& \pd^{s_1,\tau_1} \pd^{s_2,\tau_2} \psi(x) - \lsb \pd^{s_1,\tau_1} \psi(x) \rsb \lb \pd^{s_2,\tau_2} 1 \rb \\
& & - \lsb \pd^{s_2,\tau_2} \psi(x) \rsb \lb \pd^{s_1,\tau_1} 1 \rb + \psi(x) \lb \pd^{s_1,\tau_1} 1 \rb \lb \pd^{s_2,\tau_2} 1 \rb, \nn
\ea
where we have used that constants can be pulled in front of Vladimirov derivatives. Expression \eqref{219} is manifestly $1\leftrightarrow 2$ exchange symmetric.
\end{proof}

\begin{fact}
The Vladimirov derivative of $1$ vanishes with Beta function regularization, for all nontrivial derivatives.
\end{fact}
\noindent Why:
\be
\pd^{s,\tau}_x 1 = \int \pi_{-s-1,\tau} (x).
\ee
The right-hand side of the above expression is formally divergent and needs to be regularized. Using result \eqref{polyder}, that will be explained in Section \ref{derimult} below, we~obtain
\be
\label{der1vanish}
\pd^{s,\tau}_x 1 = 0,
\ee
unless $s=0$, $\tau=1$.

\begin{fact}
If we use Beta function regularization, the regularized Vladimirov derivative is the same as the usual Vladimirov derivative. 
\end{fact}
\noindent Why: Because $\pd^{s,\tau}1 =0$, from Eq. \eqref{der1vanish} above.

\subsection{The derivative of a multiplicative character}
\label{derimult}

Let's now consider how the Vladimirov derivative acts on multiplicative characters. From the position space representation \eqref{eq24} we have
\be
\pd_x^{s_1,\tau_1} \pi_{s_2,\tau_2}(x) = \Gamma\lb \pi_{s_1+1,\tau_1} \rb  \int \pi_{s_2,\tau_2}(x')\pi_{-s_1-1,\tau_1}(x'-x).
\ee
This expression is generically divergent, so let's use Beta function regularization. Perform a variable change
\be
x'=x''x, \quad dx' = |x|dx'',
\ee
so that
\be
\pd_x^{s_1,\tau_1} \pi_{s_2,\tau_2}(x) = \sgn_{\tau_1}(-1) \Gamma\lb \pi_{s_1+1,\tau_1} \rb \pi_{s_2-s_1,\tau_1\tau_2}(x)  \int \pi_{s_2,\tau_2}(x'')\pi_{-s_1-1,\tau_1}(1-x'').
\ee
Now use Eq. \eqref{B115} and obtain
\be
\pd_x^{s_1,\tau_1} \pi_{s_2,\tau_2}(x) = \sgn_{\tau_1}(-1) \Gamma\lb \pi_{s_1+1,\tau_1} \rb \frac{\Gamma(\pi_{s_2+1,\tau_2})\Gamma(\pi_{-s_1,\tau_1})}{\Gamma(\pi_{s_2-s_1+1,\tau_1\tau_2})} \pi_{s_2-s_1,\tau_1\tau_2}(x).
\ee
Using the Gamma functional equation simplifies this to
\be
\label{polyder}
\pd_x^{s_1,\tau_1} \pi_{s_2,\tau_2}(x) = \frac{\Gamma(\pi_{s_2+1,\tau_2})}{\Gamma(\pi_{s_2-s_1+1,\tau_1\tau_2})} \pi_{s_2-s_1,\tau_1\tau_2}(x).
\ee
Eq. \eqref{polyder} has the following features:
\begin{enumerate}
\item $\pd^{0,1}_x$ acts as the identity operator, as expected from Lemma \ref{lemma2}.
\item Setting $s_2=0$, $\tau_2=1$ gives that the derivative of a constant vanishes for all $\pi_{s_1,\tau_1}\neq \pi_{0,1}$.
\item Acting with $\pd_x^{s_2+1,\tau_2}$ gives zero, unless $s_2=-1$ and $\tau_2=1$.
\item The right-hand side has a pole at $s_2=s_1$, $\tau_1=\tau_2$, unless $s_{1,2}=0$ and $\tau_{1,2}=1$.
\end{enumerate}

For an arbitrary multiplicative character $\pi$, a similar computation shows that
\be
\label{betader}
\pd_x^{s_1,\tau_1} \pi(x) = \frac{\Gamma\lb \pi\pi_{1} \rb}{\Gamma\lb \pi \pi_{-s_1+1,\tau_1} \rb} \pi(x) \pi_{-s_1,\tau_1}(x).
\ee
This is analogous to the Archimedean formula for the usual derivative.

\end{spacing}

\end{document}